\documentclass[11pt]{article}

\usepackage{amssymb,latexsym,amsmath,amsthm,enumitem}

\theoremstyle{definition}
\newtheorem{theorem}{Theorem}
\newtheorem{corollary}[theorem]{Corollary}
\newtheorem{proposition}[theorem]{Proposition}
\newtheorem{lemma}[theorem]{Lemma}
\newtheorem{definition}[theorem]{Definition}
\newtheorem{example}[theorem]{Example}
\newtheorem{notation}[theorem]{Notation}
\newtheorem{remark}[theorem]{Remark}

\usepackage{titling}

\usepackage{amssymb,latexsym,amsmath}

\newcommand{\numberset}{\mathbb}
\newcommand{\N}{\numberset{N}}
\newcommand{\Z}{\numberset{Z}}

\newcommand{\F}{\numberset{F}}

\newcommand{\mC}{\mathcal{C}}
\newcommand{\mD}{\mathcal{D}}
\newcommand{\mP}{\mathcal{P}}

\usepackage{hyperref}
\newcommand{\mA}{\mathcal{A}}

\newcommand\qbin[2]{\left[\begin{matrix} #1 \\ #2 \end{matrix} \right]_q}
\newcommand\mat[2]{\F_q^{#1 \times #2}}

\newcommand{\ma}{\mat{k}{m}}
\newcommand{\cs}{\textnormal{colsp}}

\newcommand\rk{\textnormal{rk}}
\newcommand\dg{d_\textnormal{G}}
\newcommand\rkg{\textnormal{rk}_\textnormal{G}}

\usepackage[margin=2.7cm]{geometry}

\title{\textbf{Codes Endowed With the Rank Metric}\vspace{2em}}

\author{
  Elisa Gorla \\
  {Institut de Math\'{e}matiques} \\ 
  {Universit\'{e} de Neuch\^{a}tel,
Switzerland}
\\ \texttt{elisa.gorla@unine.ch}
  \and
   \ \\ Alberto Ravagnani \\
  {Department of Electrical and Computer Engineering} \\ 
  {University of Toronto,
Canada}
\\ \texttt{ravagnani@ece.utoronto.ca}
}

\date{}

\begin{document}

\maketitle

\vspace{1em}

\abstract{We review the main results of the theory of rank-metric codes, with
emphasis on their combinatorial properties. We study their
duality theory and MacWilliams identities,
comparing in particular rank-metric codes in vector and matrix representation.
We then investigate the combinatorial structure of MRD codes and optimal anticodes in the rank metric,
describing how they relate to each other.}

\vspace{2em}

\section*{Introduction}

A $q$-ary rank-metric code is a set of matrices over $\F_q$ equipped with the rank distance, which measures
the rank of the difference of a pair of matrices. Rank-metric codes were first studied in \cite{del1} by Delsarte 
for combinatorial interest.

More recently, codes endowed with the rank metric have been re-discovered 
for error correction in the context of linear network coding, and
featured prominently in the coding theory literature.

In linear network coding, a source attempts to transmit
 information packets to multiple destinations via a network of intermediate nodes. 
 The 
 nodes compute and forward in the direction of the sinks \textit{linear functions} 
 of the received packets, rather than simply routing them. 
 In \cite{origine,origine2} it was shown that linear network coding
 achieves the optimal multicast throughput over sufficiently large alphabets.
 
 Rank-metric codes were proposed in \cite{KK1,metrics2} for end-to-end error 
 correction in noisy and adversarial networks. 
 In this context, as shown in  \cite{onmetrics}, the correction capability of a rank-metric code
 is measured by a fundamental parameter, called the \textit{minimum rank distance}
 of the code.
 
 In this work we survey the main results 
 of the mathematical theory of rank-metric codes, with emphasis on
 their combinatorial structure.
 
 In Section \ref{sec:1} we introduce the most important parameters
 of a rank-metric code, namely, the minimum distance, the weight
 distribution, and the distance distribution. We then define the trace-dual of a linear rank-metric code,
 and compare the duality theories of codes in matrix and vector representation.
 In particular, we show that the former generalizes the latter.
 
 Section \ref{sec:mw} is devoted to the duality theory of codes endowed with the rank metric.
 We study how combinatorial properties of a linear code relate to 
 combinatorial properties of the dual code. In particular, we 
 show that the weight distribution of a linear code and the weight distribution
 of its dual code determine each other via a MacWilliams-type transformation.
We also show an application of the MacWilliams identities for the rank metric
 to an enumerative combinatorics problem.

In Section \ref{sec:MRD} we study codes that have the largest possible cardinality for their parameters.
These  are called 
\textit{Maximum Rank Distance} codes ({MRD} in short) and have very remarkable properties.
We first show the existence of linear MRD codes for all choices of the parameters and of the field
size. Then we prove that the dual of a linear MRD code is MRD. Finally, we show that 
the distance distribution of a (possibly non-linear) rank-metric code
is completely determined by its parameters.

Section \ref{sec:anti} is devoted to rank-metric anticodes, i.e., sets of matrices 
where the distance between any two of them is bounded from above by a given integer.
We study how codes and anticodes relate to each other, deriving in particular an upper bound for the 
cardinality of any anticode of given parameters. We conclude the section showing that the dual of
an optimal linear anticode is an optimal anticode.

\section{Rank-metric codes}
\label{sec:1}

In the sequel $q$ denotes a fixed prime power, and $\F_q$ the finite field with $q$ elements.
Moreover, $k$ and $m$ denote positive integers with $k \le m$ without loss of generality, and
$\mat{k}{m}$ is the space of $k \times m$ matrices over $\F_q$. 
Finally, for given integers $a, b \in \N$ we denote by
$$\qbin{a}{b}$$
the $q$-ary binomial coefficient of $a$ and $b$, which counts the number of 
$b$-dimensional subspaces of an $a$-dimensional space over $\F_q$.
See e.g.  \cite[Section 1.7]{stanley} for details.

\begin{definition}
 The \textbf{rank distance} is the function $d:
 \mat{k}{m} \times \mat{k}{m} \to \N$
 defined by 
 $d(M,N)=\rk(M-N)$ for all $M,N \in \mat{k}{m}$.
\end{definition}

It is easy to check that $d$ is indeed a distance function on $\ma$.

\begin{definition}
 A (\textbf{rank-metric}) \textbf{code} over $\F_q$  is a non-empty subset
 $\mC \subseteq \mat{k}{m}$. When $|\mC|\ge 2$, the 
 \textbf{minimum distance} of $\mC$ is the positive integer 
 $$d(\mC)= \min\{d(M,N) \ | \ M,N \in \mC, \ M \neq N\}.$$
 A code $\mC$ is \textbf{linear}
 if it is an $\F_q$-linear subspace of $\mat{k}{m}$. In this case 
  its \textbf{dual code} is defined 
 as $$\mC^\perp=\{N \in \mat{k}{m} \ | \ \mbox{Tr}(MN^t) =0 \mbox{ for all } 
 M \in \mC\} \subseteq \mat{k}{m},$$ where $\mbox{Tr}(\cdot)$ denotes
  the trace of a square 
 $k\times k$ matrix.
\end{definition}

The map $(M,N) \to \mbox{Tr}(MN^t) \in \F_q$ is a scalar product on 
$\mat{k}{m}$, i.e., it is symmetric, bilinear and non-degenerate.
In particular, the dual of a linear code is a linear code of 
dimension
$$\dim(\mC^\perp)=km-\dim(\mC).$$

Other fundamental parameters of a rank-metric code are the following.

\begin{definition}
 The \textbf{weight distribution} and 
 the \textbf{distance distribution}
 of a  code 
 $\mC$ are the collections $\{W_i(\mC) \ | \ i \in \N\}$ and 
 $\{D_i(\mC) \ | \ i \in \N\}$ respectively, where
 $$W_i(\mC)=|\{M \in \mC \ | \ \rk(M)=i \}|, \ \ \ \ \  
 D_i(\mC)=1/|\mC| \cdot |\{(M,N) \in \mC^2 \ | \ d(M,N)=i \}|$$ 
 for all $i \in \N$.
 \end{definition}

If $\mC$ is a linear code, then 
for all $P \in \mC$ there are precisely 
$|\mC|$ pairs $(M,N) \in \mC^2$ such that $M-N=P$. Therefore 
$$
 D_i(\mC) = 1/|\mC| \cdot \sum_{\substack{P \in \mC \\ 
 \textnormal{rk}(P)=i}} |\{ (M,N) \in \mC^2 \ | \ M-N=P\}| = W_i(\mC)
$$
for all $i \in \N$.
Moreover, if $|\mC| \ge 2$ then $d(\mC)= \min \{\rk(M) \ | \ M \in \mC, \ M \neq 0\}$.

In \cite{gabid}, Gabidulin proposed independently a different notion of
rank-metric code, in which the codewords are vectors with entries from an extension field
$\F_{q^m}$ rather than matrices over $\F_q$.

\begin{definition}
 The \textbf{rank} of a vector $v=(v_1,...,v_k)  \in \F_{q^m}^k$
 is the dimension of the linear spaces generated over $\F_q$
 by its entries, i.e., $\rkg(v)=\mbox{dim}_{\F_q} \langle v_1,...,v_k \rangle$.
 The \textbf{rank distance} between vectors $v,w \in \F_{q^m}^k$ is 
 $\dg(v,w)=\rkg(v-w)$.
 \end{definition}

One can check that $\dg$ is a distance function on 
 $\F_{q^m}^k$.
 
\begin{definition}
 A \textbf{vector rank-metric code} over $\F_{q^m}$  is a non-empty subset
 $C \subseteq \F_{q^m}^k$. When $|C|\ge 2$, the 
 \textbf{minimum distance} of $C$ is the positive integer 
 $$\dg(C)= \min\{\dg(v,w) \ | \ v,w \in C, \ v \neq w\}.$$
 The code $C$ is \textbf{linear}
 if it is an $\F_{q^m}$-linear subspace of $\F_{q^m}^k$.
 In this case the \textbf{dual} of $C$
 is defined as 
 \begin{equation*} \label{defDG}
 C^\perp=\left\{w \in \F_{q^m}^k \ | \ 
 \sum_{i=1}^k v_iw_i =0 \mbox{ for all } 
v \in C\right\} \subseteq \F_{q^m}^k.
 \end{equation*}
\end{definition}

The map $(v,w) \mapsto \sum v_iw_i$ is an $\F_{q^m}$-scalar product 
on $\F_{q^m}^k$. Therefore for all linear 
vector rank-metric codes $C \subseteq \F_{q^m}^k$ we have 
$$\dim_{\F_{q^m}} (C^\perp) = k-\dim_{\F_{q^m}}(C).$$
 
\begin{definition}
The \textbf{weight distribution} and 
 the \textbf{distance distribution}
 of a  vector rank-metric code 
 $C$ are the integer vectors  $(W_i(C) \ | \ i \in \N)$ and 
 $(D_i(C) \ | \ i \in \N)$ respectively, where
 $$W_i(C)=|\{v \in C \ | \ \rkg(v)=i \}|, \ \ \ \ \  
 D_i(C)=1/|C| \cdot |\{(v,w) \in C^2 \ | \ \dg(v,w)=i \}|$$ 
 for all $i \in \N$.
\end{definition}

There exists a natural way to associate to a 
vector rank-metric code a
code in matrix representation with the
same cardinality and metric properties.

\begin{definition}\label{mtxassoc}
 Let $\Gamma= \{ \gamma_1,...,\gamma_m\}$ be a basis
of $\F_{q^m}$ over $\F_q$. The matrix \textbf{associated} to a vector
$v \in \F_{q^m}^k$ with respect to $\Gamma$ is the $k \times m$ 
matrix 
$\Gamma(v)$ with entries in $\F_q$ defined by 
$$v_i= \sum_{j=1}^m
\Gamma(v)_{ij}\gamma_j \ \ \ \ \  \mbox{for all $i=1,...,k$}.$$
 The rank-metric code
\textbf{associated} to a vector rank-metric code $C \subseteq \F_{q^m}^k$ with respect to $\Gamma$
 is $$\Gamma(C)=\{ \Gamma(v) \ | \  v \in C\}
\subseteq \ma.$$
\end{definition}

Notice that in the previous definition the $i$-th row of
 $\Gamma(v)$ is 
the expansion
of the entry $v_i$ over the basis $\Gamma$. 

The proof of the following result is standard and 
left to the reader.

\begin{proposition}\label{gabtodel}
 For every $\F_q$-basis $\Gamma$ of 
 $\F_{q^m}$ the map 
$v \mapsto \Gamma(v)$ is an $\F_q$-linear bijective isometry 
$(\F_{q^m}^k,\dg) \to (\ma,d)$.

In particular, if $C \subseteq \F_{q^m}^k$ is a vector rank-metric code,
then $\Gamma(C)$ has the same cardinality, rank distribution and 
distance distribution as $C$.
Moreover, if $|C| \ge 2$ then 
$\dg(C)=d(\Gamma(C))$.
\end{proposition}

In the remainder of the section we compare the duality theories of 
matrix and vector rank-metric codes, showing that the former generalizes the latter.
The following results appear in~\cite{albrank}.

Given an $\F_{q^m}$-linear vector rank-metric code $C \subseteq \F_{q^m}^k$ and a basis 
$\Gamma$ of $\F_{q^m}$
over $\F_q$,
it is natural to ask whether the codes 
$\Gamma(C^\perp)$ and 
$\Gamma(C)^\perp$ coincide or not.
The answer is negative in
general, as we show in the following example.

\begin{example}
 Let $q=3$, $k=m=2$ and
$\F_{3^2}=\F_3[\eta]$, where
$\eta$ is a root of the irreducible primitive polynomial
 $x^2+2x+2 \in \F_3[x]$.
Let $\xi=\eta^2$, so that $\xi^2+1=0$.
Set $\alpha=(\xi,2)$, and let $C \subseteq \F_{3^2}^2$ be the 1-dimensional
vector rank-metric code generated by $\alpha$ over $\F_{3^2}$.  Take 
$\Gamma= \{ 1,\xi\}$
as 
basis of $\F_{3^2}$ over $\F_3$. One can 
check that $\Gamma(C)$
is generated over $\F_3$ by the two matrices 
$$\Gamma(\alpha)= \begin{bmatrix} 
                  0 & 1 \\ 2 & 0 
                 \end{bmatrix}, \ \ \ \ \ \ \  \Gamma(\xi \alpha)=
\begin{bmatrix} 
                  -1 & 0 \\ 0 & 2 
                 \end{bmatrix}.
                 $$
Let $\beta=(\xi,1) \in \F_{3^2}^2$. We have 
$\alpha_1 \beta_1 + \alpha_2 \beta_2 = 1 \neq 0$, and so $\beta \notin C^\perp$. It follows 
$\Gamma(\beta) \notin \Gamma(C^\perp)$.
On the other
hand, 
$$\Gamma(\beta)=
\begin{bmatrix} 
                  0 & 1 \\ 1 & 0 
                 \end{bmatrix},$$
and it is easy to see that $\Gamma(\beta)$ 
is trace-orthogonal to both
$\Gamma(\alpha)$ and $\Gamma(\xi \alpha)$. Therefore
$\Gamma(\beta) \in
\Gamma(C)^\perp$, hence $\Gamma(C)^\perp \neq
\Gamma(C^\perp)$.
\end{example}

Although the duality notions for matrix and vector rank-metric codes do not coincide, there is a simple
relation between them via orthogonal bases of finite fields.

 Let $\mbox{Trace}:\F_{q^m} \to \F_q$ be the map
defined by
$\mbox{Trace}(\alpha)= \alpha +\alpha^q+ \cdots + \alpha^{q^{m-1}}$ for all
 $\alpha \in \F_{q^m}$.  Bases $\Gamma= \{ \gamma_1,...,\gamma_m\}$ and 
$\Gamma'= \{ \gamma'_1,...,\gamma'_m\}$ 
of $\F_{q^m}$ over
$\F_q$ are called \textbf{orthogonal} if
 $\mbox{Trace}(\gamma'_i\gamma_j)=\delta_{ij}$ for all $i,j \in \{ 1,...,m\}$.
It is well-known that 
every basis $\Gamma$ of $\F_{q^m}$ over
$\F_q$ has a unique orthogonal basis $\Gamma'$
(see \cite{niede}, page 54).

\begin{theorem} \label{self}
Let $C \subseteq \F_{q^m}^k$ be an $\F_{q^m}$-linear 
vector rank-metric code, and let 
$\Gamma$, $\Gamma'$
be orthogonal bases of $\F_{q^m}$ over $\F_q$. We have
$$\Gamma'(C^\perp)= \Gamma(C)^\perp.$$
In particular, $C$ has the same weight 
distribution as $\Gamma(C)$, and $C^\perp$
has the same weight distribution as $\Gamma(C)^\perp$.
\end{theorem}

\begin{proof} Write $\Gamma=\{ \gamma_1,...,\gamma_m\}$ and $\Gamma'=\{
\gamma'_1,...,\gamma'_m\}$. 
 Let $M \in \Gamma'(C^\perp)$ and $N \in
\Gamma(C)$. There exist $\alpha \in C^\perp$ and $\beta \in C$
 such that
$M=\Gamma'(\alpha)$ and $N=\Gamma(\beta)$. By 
Definition~\ref{mtxassoc} we have
\begin{equation} 
 0 \ = \ \sum_{i=1}^k
\alpha_i\beta_i\ = \ \sum_{i=1}^k \sum_{j=1}^m M_{ij}\gamma'_j \sum_{t=1}^m
N_{it} \gamma_t \ = \ \sum_{i=1}^k \sum_{j=1}^m \sum_{t=1}^m
M_{ij}N_{it} \gamma'_j \gamma_t.\label{eqqq}
\end{equation}
Applying the function $\mbox{Trace}:\F_{q^m} \to \F_q$ to both
sides of
equation (\ref{eqqq}) we obtain
\begin{equation*}
 0 = \mbox{Trace} \left( \sum_{i=1}^k \sum_{j=1}^m \sum_{t=1}^m
M_{ij}N_{it} \gamma'_j \gamma_t \right) 
 =  \sum_{i=1}^k \sum_{j=1}^m \sum_{t=1}^m
M_{ij}N_{it} \mbox{Trace}(\gamma'_j \gamma_t)  =
 \mbox{Tr}(MN^t). 
\end{equation*}
Therefore $\Gamma'(C^\perp)
\subseteq
\Gamma(C)^\perp$. Proposition \ref{gabtodel} 
implies that 
$\Gamma'(C^\perp)$ and $\Gamma(C)^\perp$ have the same dimension over
$\F_q$.
Hence the two codes are equal. The second part of the statement follows
from Proposition \ref{gabtodel}.
\end{proof}

Theorem \ref{self} shows that the duality theory of $\F_q$-linear rank-metric codes in matrix representation 
can be 
regarded as a
generalization of the duality theory of $\F_{q^m}$-linear 
vector rank-metric codes. 
For this reason, in the sequel we only treat 
rank-metric codes in matrix representation.

\section{MacWilliams identities for the rank metric} \label{sec:mw}

This section is devoted to the duality theory of codes endowed with the rank metric.
We concentrate on linear rank-metric codes, and show that the weight distributions
of a code $\mC$ and its dual code $\mC^\perp$ determine each other via a 
MacWilliams-type transformation. This result was established by 
Delsarte in \cite[Theorem 3.3]{del1}
using the machinery of association schemes, and may be regarded as the 
rank-analogue of a celebrated theorem by MacWilliams on the weight
distribution of linear codes endowed with the Hamming metric (see
\cite{MWhamming}). In this section we present a lattice-theoretic proof 
inspired by \cite[Theorem 27]{alblattic}.

\begin{notation}
 We denote  by $\cs(M) \subseteq \F_q^k$ the $\F_q$-space generated by the columns of 
 a matrix $M \in \mat{k}{m}$. Given a code $\mC \subseteq \ma$ and an $\F_q$-subspace 
 $U \subseteq \F_q^k$, we let
 $$\mC(U)=\{M \in \mC  \ | \ \cs(M) \subseteq U\} \subseteq \ma$$
 be the set of matrices in $\mC$ whose columnspace is contained in $U$.
 \end{notation}
 
Note that for all $M,N \in \ma$ we  have
$\cs(M+N) \subseteq \cs(M)+\cs(N)$. As a
consequence, 
if $U \subseteq \F_q^k$ is an $\F_q$-linear subspace and $\mC \subseteq \ma$ is a linear code, then 
$\mC(U)$
is a linear code as well.

We start with a series of preliminary results. In the sequel we denote by $U^\perp$ the 
orthogonal (or dual) of an $\F_q$-vector space 
$U \subseteq \F_q^k$ with respect to the standard inner product of 
$\F_q^k$. It will be clear from  context if by ``$\perp$'' we denote the
trace-dual in $\ma$ or the standard dual in $\F_q^k$.

\begin{lemma} \label{dimensioni}
 Let  $U \subseteq \F_q^k$ be a subspace.  The following hold.
 \begin{enumerate}
\item $\dim
(\ma(U))=m \cdot \dim(U)$. \label{dimensioni1}
\item $\ma(U)^\perp = \ma(U^\perp)$. \label{dimensioni2}
 \end{enumerate}
\end{lemma}
\begin{proof}
\begin{enumerate}
 \item Let $s=\dim(U)$ and $V=\{
(x_1,...,x_k) \in \F_q^k \ | \ x_i=0 \mbox{ for } i >s \} \subseteq \F_q^k$. There exists an
$\F_q$-isomorphism $g:\F_q^k \to \F_q^k$ that maps $U$ to $V$. Let $G \in
\mat{k}{k}$ be the invertible matrix associated to $g$
with
respect to the canonical
basis $\{e_1,...,e_k\}$ of $\F_q^k$, i.e.,
$$g(e_j) = \sum_{i=1}^k G_{ij} e_i \ \ \ \ \ \mbox{ for all }
j=1,...,k.$$ 
The map
$M \mapsto GM$ is an $\F_q$-isomorphism
$\ma(U) \to \ma(V)$. Property \ref{dimensioni1} of the lemma now directly follows from the definition 
of $\ma(V)$. 

\item  Let $N \in \ma(U^\perp)$ and $M \in \ma(U)$. Using the
definition of trace-product one sees that 
$\mbox{Tr}(MN^t) =\sum_{i=1}^m  \langle M_i, N_i \rangle$, where $\langle \cdot
, \cdot \rangle$ is the standard inner product of $\F_q^k$, and $M_i$, $N_i$
denote the $i$-th column of $M$ and $N$ (respectively). Each
 column of $N$ belongs to $U^\perp$, and
each column of $M$ belongs to
$U$. Therefore $\mbox{Tr}(MN^t)=0$, hence
$\ma(U^\perp) \subseteq \ma(U)^\perp$. By property 
\ref{dimensioni1}, the two spaces 
$\ma(U^\perp)$ and $\ma(U)^\perp$ have the same
dimension over $\F_q$. Therefore they are
equal.
\end{enumerate}
\end{proof}

The following result is \cite[Lemma 28]{albrank}.

\begin{proposition} \label{tecn}
 Let $\mC \subseteq \ma$ be a linear code, and let $U
\subseteq \F_q^k$ be a subspace of dimension $u$ over $\F_q$. Then 
 $$|\mC(U)|= \frac{|\mC|}{ q^{m(k-u)}}
|\mC^\perp(U^\perp)|.$$
\end{proposition}
\begin{proof}
 We have  
 $\mC(U)^\perp = (\mC \cap \ma(U))^\perp=\mathcal{\mC}^\perp
+\ma(U)^\perp=\mC^\perp + \ma(U^\perp)$,
where the last equality follows from part \ref{dimensioni2} of 
Lemma \ref{dimensioni}. Therefore
 \begin{equation} \label{eq2}
 |\mC(U)| \cdot |\mC^\perp +
\ma(U^\perp)|=q^{km}.
 \end{equation}
 On the other hand, part \ref{dimensioni1} of Lemma \ref{dimensioni} gives
 $$\dim (\mC^\perp + \ma(U^\perp))=\dim (\mC^\perp)+m \cdot \dim
(U^\perp)-\dim  (\mC^\perp(U^\perp)).$$
As a consequence, 
  \begin{equation}\label{eq3}
|\mC^\perp + \ma(U^\perp)| = \frac{q^{km} \cdot
q^{m(k-u)}}{|\mC| \cdot
|\mC^\perp(U^\perp)|}.
 \end{equation}
 Combining equations (\ref{eq2}) and (\ref{eq3}) one obtains the
proposition.
\end{proof}

We will  also need the following preliminary lemma, which is
an explicit version of the M\"{o}bius inversion formula 
for the lattice of subspaces of $\F_q^k$.
We include a short proof for completeness.
See  \cite[Sections 3.7 -- 3.10]{stanley} for details.

\begin{lemma} \label{mobb}
 Let $\mP(\F_q^k)$ be the set of all $\F_q$-subspaces of $\F_q^k$, and let 
 $f:\mP(\F_q^k) \to \Z$ be any function. Define $g:\mP(\F_q^k) \to \Z$ by
 $g(V)=\sum_{U \subseteq V} f(U)$ for all $V \subseteq \F_q^k$. Then for all $i
\in \{0,...,k\}$ and for any subspace $V \in \mP(\F_q^k)$ with $\dim(V)=i$
we have 
 $$f(V)= \sum_{u=0}^i {(-1)}^{i-u} q^{\binom{i-u}{2}}  \sum_{\substack{U
\subseteq V \\ \dim(U)=u}} g(U).$$
\end{lemma}

\begin{proof}
Fix an integer $i \in \{0,...,k\}$ and a vector space $V \in \mP(\F_q^k)$ with
$\dim(V)=i$. We inductively define a function $\mu:\{U \in \mP(\F_q^k) \ | \ 
U\subseteq V\} \to
\Z$  by 
$\mu(U)=1$ if $U=V$, and 
$\mu(U)=-\sum_{U \subsetneq S \subseteq V} \mu(S)$ if $U \subsetneq V$.
By definition of $g$ we have 

$$ \sum_{U\subseteq V} \mu(U) g(U) = \sum_{U\subseteq V} \mu(U) \sum_{S
\subseteq U} f(S) = \sum_{S \subseteq V} f(S) \sum_{S \subseteq U \subseteq V}
\mu(U)=f(V),$$
where the last equality immediately follows from the definition of $\mu$.
Therefore it suffices to show that for all $U \subseteq V$ we have 
\begin{equation}\label{fmob}
 \mu(U)= {(-1)}^{i-u} q^{\binom{i-j}{2}},
\end{equation}
where $u=\dim(U)$. We proceed by induction on $i-u$.  If $i=u$
then equation (\ref{fmob}) is trivial. Now assume $i>u$. By definition
of $\mu$ and the induction hypothesis we have 
\allowdisplaybreaks
\begin{eqnarray*}
 \mu(U) \ = \ - \sum_{U \subsetneq S \subseteq V}
\mu(S) &=&  -\sum_{s=u+1}^i {(-1)}^{i-s} q^{\binom{i-s}{2}} \qbin{i-j}{s-u} \\
&=&  -\sum_{s=u+1}^i {(-1)}^{i-s} q^{\binom{i-s}{2}} \qbin{i-u}{i-s} \\
&=&  -\sum_{s=0}^{i-u} {(-1)}^{s} q^{\binom{s}{2}} \qbin{i-u}{s} +
{(-1)}^{i-u} q^{\binom{i-u}{2}} \\
&=& {(-1)}^{i-u} q^{\binom{i-u}{2}},
\end{eqnarray*}
where the last equality follows from the $q$-Binomial Theorem (see 
\cite{stanley}, page 74).
\end{proof}

We can now prove the main result of this section,  
first established by Delsarte in \cite[Theorem~3.3]{del1}.
A proof for the special case of 
$\F_{q^m}$-linear vector rank-metric codes using different techniques
can be found in \cite{gadu}.

\begin{theorem}[MacWilliams identities for the rank metric] \label{mwid}
 Let $\mC \subseteq \ma$ be an linear rank-metric code. For all 
 $i \in \{0,...,k\}$ we have
 $$W_i(\mC^\perp) = \frac{1}{|\mC|} 
 \sum_{j=0}^k W_j(\mC) \sum_{u=0}^k {(-1)}^{i-u} q^{mu+
 \binom{i-u}{2}} \qbin{k-u}{k-i}\qbin{k-j}{u}.$$
\end{theorem}

\begin{proof}
For all subspaces $V \subseteq \F_q^k$ define 
$$f(V)= |\{ M \in \mC^\perp \ | \ \mbox{colsp}(M) = V \}|, \ \ \ \ \ \ 
g(V)= \sum_{U \subseteq V} f(U) = |\mC^\perp(V)|.$$
By Lemma \ref{mobb}, for any $i \in \{0,...,k\}$ and for any vector space $V \subseteq \F_q^k$ 
of dimension $i$ we have 
\allowdisplaybreaks
\begin{eqnarray*}
f(V) &=& \sum_{u=0}^i {(-1)}^{i-u} q^{\binom{i-u}{2}}  \sum_{\substack{U
\subseteq V \\ \dim(U)=u}} |\mC^\perp(U)| \\
&=& \sum_{u=0}^i {(-1)}^{i-u} q^{\binom{i-u}{2}}
\sum_{\substack{T
\subseteq \F_q^k \\ T \supseteq V^\perp \\ \dim(T)=k-u}}
|\mC^\perp(T^\perp)| \\
&=& \frac{1}{|\mC|} \sum_{u=0}^i {(-1)}^{i-u} q^{mu+\binom{i-u}{2}}
\sum_{\substack{T
\subseteq \F_q^k \\ T \supseteq V^\perp \\ \dim(T)=k-u}}
|\mC(T)|,
\end{eqnarray*}
where the last equality follows from Proposition 
\ref{tecn}. Now observe that
\allowdisplaybreaks
\begin{eqnarray} 
W_i(\mC^\perp) &=& \sum_{\substack{V
\subseteq \F_q^k \\ \dim(V)=i}} f(V) \nonumber \\ &=& 
\frac{1}{|\mC|} \ \sum_{u=0}^i  {(-1)}^{i-u} q^{mu+\binom{i-u}{2}}  
  \sum_{\substack{V
\subseteq \F_q^k \\ \dim(V)=i}} \sum_{\substack{T
\subseteq \F_q^k \\ T \supseteq V^\perp \\ \dim(T)=k-u}}
|\mC(T)| \nonumber \\ &=&
\frac{1}{|\mC|} \ \sum_{u=0}^i  {(-1)}^{i-u} q^{mu+\binom{i-u}{2}}
\sum_{\substack{T
\subseteq \F_q^k \\ \dim(T)=k-u}} 
\sum_{\substack{V
\subseteq \F_q^k \\ V \supseteq T^\perp \\ \dim(V)=i}} |\mC(T)|  \nonumber  \\
&=& \frac{1}{|\mC|} \ \sum_{u=0}^i  {(-1)}^{i-u} q^{mu+\binom{i-u}{2}}
 \qbin{k-u}{i-u}  \sum_{\substack{T
\subseteq \F_q^k \\ \dim(T)=k-u}}  |\mC(T)|. \label{mw1}
\end{eqnarray}
On the other hand,
\allowdisplaybreaks
\begin{eqnarray}
\sum_{\substack{T
\subseteq \F_q^k \\ \dim(T)=k-u}}  |\mC(T)| &=& 
\sum_{\substack{T
\subseteq \F_q^k \\ \dim(T)=k-u}} \sum_{j=0}^{k-u}
\sum_{\substack{S
\subseteq T \\ \dim(S)=j}} |\{M \in \mC \ | \ \mbox{colsp}(M) = S\}| 
\nonumber \\
&=& \sum_{j=0}^{k-u} \sum_{\substack{S
\subseteq \F_q^k \\ \dim(S)=j}} \sum_{\substack{T
\subseteq \F_q^k \\ T \supseteq S \\ \dim(T)=k-u}} 
 |\{M \in \mC \ | \ \mbox{colsp}(M) = S\}|  \nonumber \\
 &=& \sum_{j=0}^{k-u} \qbin{k-j}{u} W_j(\mC). \label{mw2}
\end{eqnarray}
Combining equations (\ref{mw1}) and (\ref{mw2}) one obtains the desired result.
\end{proof}

\begin{example}
Let $q=2$, $k=2$, $m=3$. Let $\mC \subseteq \ma$ be the $2$-dimensional linear code 
generated over $\F_5 \cong \Z/5\Z$ by the matrices
$$\begin{bmatrix}
1 & 0 & 2 \\ 0 & 2 & 4 \end{bmatrix}, \ \ \ \ \begin{bmatrix}
2 & 3 & 0 \\ 1 & 4 & 0
\end{bmatrix}.$$
We have $W_0(\mC)=1$, $W_1(\mC)=8$ and $W_2(\mC)=16$. Applying Theorem \ref{mwid} one can easily compute
$W_0(\mC^\perp)=1$, $W_1(\mC^\perp)=65$ and $W_2(\mC)=560$. Observe that 
$\mC^\perp$ has dimension $6-2=4$, and that 
$1+64+560=625=5^4$, as expected.
\end{example}

We now present a different formulation 
of the MacWilliams identities for the rank metric.
The following result is \cite[Theorem 31]{albrank}.

\begin{theorem} \label{theo}
 Let $\mathcal{C} \subseteq \ma$ be a linear code. 
 For all $0 \le \nu \le k$ we have 
$$\sum_{i=0}^{k-\nu} W_i(\mC) \qbin{k-i}{\nu} \ = \ \ 
\frac{|\mC|}{q^{m\nu}} \; \sum_{j=0}^{\nu} W_j(\mC^\perp) \qbin{k-j}{\nu-j}.$$
\end{theorem}

\begin{proof}
Proposition \ref{tecn} gives
\begin{equation} \label{mmww1}
\sum_{\substack{U \subseteq \F_q^k \\ \dim(U)=k-\nu}} |\mC(U)| \ \ \ = 
\ \ \  \frac{|\mC|}{q^{m\nu}}  \sum_{\substack{U \subseteq \F_q^k \\ \dim(U)=\nu}} |\mC^\perp(U)|.
\end{equation} 
Observe that
\allowdisplaybreaks
\begin{eqnarray}
\sum_{\substack{U \subseteq \F_q^k \\ \dim(U)=k-\nu}} |\mC(U)| &=&
|\{(U,M) \ | \ U \subseteq \F_q^k, \ \dim(U)=k-\nu, \ M \in \mC, \ \mbox{colsp}(M) \subseteq U \}| 
\nonumber \\
&=& \sum_{M \in \mC} |\{ U \subseteq \F_q^k, \ \dim(U)=k-\nu, \
\mbox{colsp}(M)
\subseteq U \}| \nonumber  \\
&=&  \sum_{i=0}^k \sum_{\substack{M \in \mC \\ \mbox{\small{rk}}(M)=i}} |\{ U
\subseteq \F_q^k, \ \dim(U)=k-\nu, \ \mbox{colsp}(M)
\subseteq U \}| \nonumber \\
&=&  \sum_{i=0}^k \sum_{\substack{M \in \mC \\ \mbox{\small{rk}}(M)=i}}
\qbin{k-i}{k-\nu-i} \; = \; \sum_{i=0}^{k-\nu} W_i(\mC)
\qbin{k-i}{\nu}. \label{11}
\end{eqnarray}
Using the same argument with $\mC^\perp$ and $k-\nu$ one shows that
\begin{equation}
\sum_{\substack{U \subseteq \F_q^k \\ \dim(U)=\nu}} |\mC^\perp(U)|  \ = \ \sum_{j=0}^{k-\nu} W_j(\mC^\perp)
\qbin{k-j}{\nu-j}. \label{22}
\end{equation}
The result now follows combining equations (\ref{mmww1}), (\ref{11}) and (\ref{22}).
\end{proof}

\begin{remark}
The two formulations of the MacWilliams identities for the rank metric
given in Theorems \ref{mwid} and \ref{theo} are equivalent.
See \cite[Corollary 1 and Proposition 3]{gadu} and \cite[Theorem 61]{albrank}
for details. 
\end{remark}

The next theorem is \cite[Theorem 27]{paperdistr}, and shows that the weight distribution
of a linear code is determined by its parameters, together with 
the number of codewords of small weight. 
We state it without proof.
An application of this result will be given
in  Section \ref{sec:MRD} (see Corollary~\ref{coroqmrd}).

\begin{theorem}\label{04-02-15}
Let $\mC \subseteq \ma$ be a linear code with $1 \le \dim(\mC) \le km-1$, minimum distance 
$d=d(\mC)$, and dual minimum distance $d^\perp=d(\mC^\perp)$.
Let $\varepsilon=1$ if $\mC$ is MRD, and $\varepsilon=0$ otherwise.
For all $1 \le i \le d^\perp$ we have
\begin{eqnarray*}
W_{k-d^\bot+i}(\mC) &=& (-1)^{i}q^{i \choose 2}
\sum\limits_{u=d^\bot}^{k-d} \qbin{u}{d^\bot-i}  \qbin{u-d^\bot+i-1}{i-1}
W_{k-u}(\mC) \\
&+& \qbin{k}{d^\bot-i}\sum\limits_{u=0}^{i-1-\varepsilon} {(-1)}^uq^{u \choose 2
} \qbin{k-d^\bot+i}{u}    \left(q^{\dim(\mC)-m(d^\perp-i+u)}-1\right).
\end{eqnarray*}
In particular, $k$, $m$, $t$, $d$, $d^\perp$ and $W_d(\mC),\ldots, W_{k-d^\perp}(\mC)$
completely determine the weight distribution of $\mC$.
\end{theorem}

We conclude this section showing how MacWilliams identities
for the rank metric can be employed to solve certain enumerative problems of matrices
over finite fields. The following result is \cite[Corollary 52]{alblattic}.

\begin{corollary}
Let $I \subseteq \{ (i,j)  \in \{1,...,k\} \times \{1,...,m\} \  | \ i=j \}$ be a set of diagonal
entries.
For all $0 \le r \le k$ the number of $k \times m$ matrices $M$ over $\F_q$
having rank $r$
and $M_{ij}=0$ for all $(i,j) \in I$ is
$$q^{-|I|} \ \sum_{t=0}^{k} 
\binom{|I|}{t} (q-1)^t
 \sum_{u=0}^k (-1)^{r-u} \
q^{mu+\binom{r-u}{2}}  \ 
\qbin{k-u}{k-r} \qbin{k-t}{u}.$$
\end{corollary}

\begin{proof}
Define the linear code $\mC=\{M \in \ma \ | \ M_{ij}=0 \mbox{ for all } (i,j) \notin I\} \subseteq \ma$.
Then $\dim(\mC)=|I|$, $W_t(\mC)=0$ for $|I| < t \le k$, and 
$$W_t(\mC)= \binom{|I|}{t}{(q-1)}^t$$
for $0 \le t \le |I|$. Moreover, 
$\mC^\perp = \{M \in \ma \ | \ M_{ij}=0 \mbox{ for all } (i,j) \in I\}$.
Therefore the number of matrices $M \in \ma$ having rank $r$
and $M_{ij}=0$ for all $(i,j) \in I$ is
precisely $W_r(\mC^\perp)$. The corollary now follows from Theorem \ref{mwid}.
\end{proof}

\section{MRD codes} \label{sec:MRD}

In this section we study rank-metric codes that have the largest possible cardinality for their
parameters. We start with a Singleton-type bound for the cardinality of a rank-metric code
of given minimum distance. A code is called MRD if
it attains the bound. We then show that for any admissible choice of the parameters 
there exists a linear MRD code with those parameters.

In the second part of the section we study general structural properties of MRD codes.
We first prove in Theorem \ref{dualofis} that the dual of a linear MRD
code is MRD. Then we show in Theorem \ref{weight0} that the weight distribution of a possibly non-linear
MRD code $\mC \subseteq \ma$  with $0 \in \mC$ is determined by
$k$, $m$ and $d(\mC)$. As a corollary, we prove that these three
parameters completely determine the distance distribution of any MRD code. Our proofs are inspired by the lattice-theory
approach to the weight functions of coding theory proposed in \cite{tesialberto} and 
\cite{alblattic}.

\begin{theorem}[Singleton-like bound] \label{singbound}
Let $\mC \subseteq \ma$ be a rank-metric code with $|\mC| \ge 2$ and minimum distance 
$d$. Then 
$|\mC| \le q^{m(k-d+1)}$.
\end{theorem}

\begin{proof}
Let $\pi: \mC \to \mat{(k-d+1)}{m}$ denote the projection on the last 
$k-d+1$ rows. Since $\mC$ has minimum distance $d$,
the map $\pi$ is injective. Therefore
$$|\mC| = |\pi(\mC)| \le q^{m(k-d+1)}.$$
\end{proof}

A code is MRD if its parameters attain the Singleton-like bound.

\begin{definition}
We say that $\mC \subseteq \ma$ is an \textbf{MRD}  code 
if $|\mC|=1$, or 
$|\mC| \ge 2$ and $|\mC|=q^{m(k-d+1)}$, where $d=d(\mC)$.
\end{definition}

We now prove that for any choice of $q$, $k$,
$m$ and $d$ there exists a linear rank-metric code 
$\mC \subseteq \ma$ that attains the bound of
Theorem \ref{singbound}. This result was first shown by Delsarte in \cite{del1}, and rediscovered 
independently by Gabidulin in 
\cite{gabid} and by K\"{o}tter and Kschischang in \cite{KK1} in the 
context of linear network coding.

\begin{theorem} \label{attaingabid}
For all $1 \le d \le k$ there exists an $\F_{q^m}$-linear
 vector rank-metric code 
 $C \subseteq \F_{q^m}^k$ with $\dg(C)=d$ 
 and $\dim_{\F_{q^m}}(C)=k-d+1$. In particular, there 
 exists a linear MRD code 
 $\mC \subseteq \ma$ with $d(\mC)=d$.
\end{theorem}

We include an 
elegant proof for 
Theorem \ref{attaingabid} from
\cite{KK1}. 
Recall that a \textbf{linearized polynomial} $p$ over $\F_{q^m}$ is a polynomial of the
form
$$p(x)=\alpha_0x + \alpha_1 x^q+ \alpha_2 x^{q^2}+ \cdots +\alpha_{s}x^{q^s},
\ \ \ \ \ \alpha_i \in \F_{q^m}, \ \ i=0,...,s.$$
The \textbf{degree} of $p$, denoted by $\deg(p)$,  is the largest integer $i \ge
0$
such that $\alpha_i \neq 0$. The $\F_{q^m}$-vector space of linearized
polynomials
over $\F_{q^m}$ of degree at most $s$ is denoted by 
$\mbox{Lin}_q(m,s)$. It is easy to see that $\dim_{\F_{q^m}} (\mbox{Lin}_q(m,s))
= s+1$.

\begin{remark} \label{rootspol}
The roots
of a linearized polynomial $p$ over $\F_{q^m}$ form an $\F_q$-vector subspace of
$\F_{q^m}$ (see \cite{niede}, Theorem 3.50),
which we denote by
$V(p) \subseteq \F_{q^m}$ in the sequel. Clearly, for any
non-zero linearized 
polynomial $p$ we have
$\dim_{\F_q}V(p) \le \deg(p)$ by the Fundamental Theorem of Algebra.
\end{remark}

\begin{proof}[Proof of Theorem \ref{attaingabid}]
Let 
$E=\{\beta_1,...,\beta_k\} \subseteq \F_{q^m}$ be a set of $\F_q$-independent
elements. These elements exist
as $k \le m$ by assumption.
Define the $\F_{q^m}$-linear map 
$$\mbox{ev}_E: \mbox{Lin}_q(m,k-d) \to
\F_{q^m}^k, \ \ \ \ 
\mbox{ev}_E(p)=(p(\beta_1),...,p(\beta_k)) \ \mbox{for  $p \in
\mbox{Lin}_q(m,k-d)$}.$$
We claim that 
$C=\mbox{ev}_E(\mbox{Lin}_q(m,k-d)) \subseteq \F_{q^m}^k$ is a 
vector rank-metric code with the desired properties. 

Clearly, $C$ is 
$\F_{q^m}$-linear.
Now let $p \in \mbox{Lin}_q(m,k-d)$ be a non-zero linearized polynomial, and
let 
$W \subseteq \F_{q^m}$ denote the space generated over $\F_q$ by the evaluations
 $p(\beta_1),...,p(\beta_k)$.
The polynomial $p$ induces an $\F_q$-linear evaluation map 
$p:\langle \beta_1,...,\beta_k \rangle_{\F_q} \to \F_{q^m}$. The image of 
$p$ is $W$, and therefore by the rank-nullity theorem we have 
$\dim_{\F_q}(W)=k-\dim_{\F_q}V(p)$. By Remark \ref{rootspol}
we conclude $\dim_{\F_q}(W) \ge k-(k-d)=d$. This shows  that 
 $\dg(C) \ge d$. In particular, as $d \ge 1$, the map $\mbox{ev}_E$ is
injective, and  the dimension of 
$C$ is $\dim_{\F_{q^m}}(C)=k-d+1$.
Combining Proposition~\ref{gabtodel} and Theorem \ref{singbound}
we obtain $\dg(C)=d$.

The second part of the theorem immediately follows from 
Proposition \ref{gabtodel}.
\end{proof}

The MRD code construction in the 
proof of Theorem \ref{attaingabid} was later generalized by 
Sheekey in~\cite{sh06}, introducing a new class of MRD codes.

The reminder of the section is devoted to the structural properties of 
MRD codes.  We start with a preliminary result from \cite[Chapter 7]{tesialberto}.

\begin{lemma} \label{design}
Let $\mC \subseteq \ma$ be an MRD code with 
$|\mC| \ge 2$ and minimum distance $d$. For all subspaces 
$U \subseteq \F_q^k$ with $u=\dim(U) \ge d-1$ we have 
$$|\mC(U)| =q^{m(u-d+1)}.$$
\end{lemma}

\begin{proof}
As in Lemma \ref{dimensioni}, define the space 
 $V=\{(x_1,...,x_k) \in \F_q^k \ | \ x_i =0 \mbox{ for 
$i > u$} \} \subseteq \F_q^k$. Let $g: \F_q^k \to \F_q^k$ be an 
$\F_q$-isomorphism with $f(U)=V$. Denote by 
$G \in \mat{k}{k}$ the matrix associated to $g$ with respect to
the canonical basis of $\F_q^k$. Define the rank-metric code $\mD=G\mC=\{GM \ | \ M \in \mC\}$.
Clearly, $\mD$ has the same dimension and minimum distance as 
$\mC$. In particular, it is MRD. Observe moreover that $\mC(U)=\mD(V)$.

Now consider the maps
$$\mD \stackrel{\pi_1}{\longrightarrow} \mat{(k-d+1)}{m}
\stackrel{\pi_2}{\longrightarrow} \mat{(k-u)}{m},$$
where $\pi_1$ is the projection on the 
last $k-d+1$ coordinates, and  
$\pi_2$ is the projection on the last 
$k-u$ coordinates. Since 
$d(\mD)=d$, $\pi_1$ is injective. Since
$\mD$ is MRD, we have $\log_q(|\mD|)= m(k-d+1)$. Therefore
$\pi_1$ is bijective. The map $\pi_2$ is $\F_q$-linear and surjective. Therefore 
$$|\pi_2^{-1}(0)| = |\pi_2^{-1}(M)| = q^{m(u-d+1)}
 \ \mbox{ for all $M \in \mat{(k-u)}{m}$}.$$
  Since 
 $\pi_1$ is bijective and $\pi_2$ is surjective, the map $\pi=\pi_2 \circ \pi_1$ is 
 surjective. Moreover,  
 $$|\pi^{-1}(0)| = |\pi^{-1}(M)| = q^{m(u-d+1)}
 \ \mbox{ for all $M \in \mat{(k-u)}{m}$}.$$
 The lemma now follows from the identity 
 $\mC(U)=\mD(V)=\pi^{-1}(0)$.
 \end{proof}
 
 We can now show that the dual of a linear MRD code is MRD.
 The next fundamental result is \cite[Theorem 5.5]{alblattic}.
 
 \begin{theorem} \label{dualofis}
 Let $\mC \subseteq \ma$ be a linear MRD code.
 Then $\mC^\perp$ is MRD.
 \end{theorem}
 
 \begin{proof}
 The result is immediate if $\dim(\mC) \in \{0,km\}$. Assume $1 \le \dim(\mC) \le km-1$, and let 
 $d=d(\mC)$, $d^\perp=d(\mC^\perp)$. Applying Theorem \ref{singbound} to 
 $\mC$ and $\mC^\perp$ we obtain
 $$\dim(\mC)\le m(k-d+1), \ \ \ \ \ \ \dim(\mC^\perp)\le m(k-d^\perp+1).$$
 Therefore $km=\dim(\mC)+\dim(\mC^\perp) \le 2mk-m(d+d^\perp)+2m$, i.e.,
 \begin{equation} \label{eqqq}
 d+d^\perp \le k+2.
 \end{equation}
 Let $U \subseteq \F_q^k$ be any $\F_q$-subspace with
$\dim(U)=k-d+1$. 
By Proposition \ref{tecn} we have 
\begin{equation}\label{eqqq2}
|\mC^\perp(U)|= \frac{|\mC^\perp|}{q^{m(d-1)}} |\mC(U^\perp)|.
\end{equation}
Since $\dim(U^\perp)=d-1$,
by Lemma 
\ref{design} we have 
$|\mC(U^\perp)| = |\mC|/q^{m(k-d+1)}=1$,
where the last equality follows from the fact that $\mC$ is MRD.
Therefore (\ref{eqqq2}) becomes
$$|\mC^\perp(U)|= \frac{|\mC^\perp|}{q^{m(d-1)}}= 
\frac{q^{km}/q^{m(d-1)}}{q^{m(d-1)}}=1.$$
Since $U$ is arbitrary with $\dim(U)=k-d+1$, this shows $d^\perp \ge k-d+2$.
Using  (\ref{eqqq}) we conclude $d^\perp=k-d+2$.
The theorem now follows from 
$$\dim(\mC^\perp)=km-\dim(\mC)=km-m(k-d+1)=m(k-d^\perp+1).$$
\end{proof}

The proof of Theorem \ref{dualofis} also shows the following useful characterization of linear
MRD codes in terms of their minimum distance and dual minimum distance.

\begin{proposition} \label{carMRD}
Let $\mC \subseteq \ma$ be a linear code with 
$1 \le \dim(\mC) \le km-1$. The following are equivalent.
\begin{enumerate}
\item $\mC$ is MRD,
\item $\mC^\perp$ is MRD,
\item $d(\mC)+d(\mC^\perp)=k+2$.
\end{enumerate}
\end{proposition}

In the remainder of the section we concentrate on the weight and distance distributions of
(possibly non-linear) MRD codes. We start with a result 
on the weight distribution of MRD codes containing the zero vector (see \cite[Theorem 7.46]{tesialberto}).
 
 \begin{theorem} \label{weight0}
 Let $\mC$ be an MRD code with $|\mC| \ge 2$ and $0 \in \mC$. 
 Let $d=d(\mC)$. Then $W_0(\mC)=1$, $W_i(\mC)=0$ for $1 \le i \le d-1$, and 
 $$W_i(\mC) =  \sum_{u=0}^{d-1} {(-1)}^{i-u} q^{\binom{i-u}{2}}
\qbin{k}{i}\qbin{i}{u} +  \sum_{u=d}^i {(-1)}^{i-u} q^{\binom{i-u}{2}+m(u-d+1)}
\qbin{k}{i}\qbin{i}{u}$$
for 
$d \le i \le k$. 
 \end{theorem}
 
 \begin{proof}
 Since $0 \in \mC$, we have $W_0(\mC)=1$ and $W_i(\mC)=0$ for $1 \le i \le d-1$.
  For all subspaces $V \subseteq \F_q^k$ define 
$$f(V)= |\{ M \in \mC \ | \ \mbox{colsp}(M) = V \}|, \ \ \ \ \ \ 
g(V)= \sum_{U \subseteq V} f(U) = |\mC(V)|.$$
Fix $0 \le i \le k$ and a vector space 
$V \subseteq \F_q^k$ of dimension $i$. By Lemma \ref{mobb} we have 
$$f(V) = \sum_{u=0}^i {(-1)}^{i-u} q^{\binom{i-u}{2}}
\sum_{\substack{U \subseteq V \\ \textnormal{dim}(U)=u}} g(U).$$
Using Lemma \ref{design} and the fact that $\mC$ is MRD with $0 \in \mC$ we obtain
$$g(U)= \left\{ \begin{array}{cl}  1 & \mbox{ if 0 $\le \dim(U) \le d-1$,} \\
q^{m(u-d+1)} & \mbox{ if $d \le \dim(U) \le k$.}\end{array}\right.\ $$
Therefore
$$f(V) = \sum_{u=0}^{d-1} {(-1)}^{i-u} q^{\binom{i-u}{2}}
\qbin{i}{u} +  \sum_{u=d}^i {(-1)}^{i-u} q^{\binom{i-u}{2}+m(u-d+1)}
\qbin{i}{u}.$$
The result now follows from the identity 
\begin{equation*}
W_i(\mC) = \sum_{\substack{V \subseteq \F_q^k \\ \textnormal{dim}(V)=i}} f(V).\qedhere
\end{equation*} 
\end{proof}

Different formulas for the weight distribution of linear MRD codes were obtained  in
\cite{mcg} using elementary methods.

Theorem \ref{weight0} implies  the following 
\cite[Theorem 5.6]{del1}, which states 
that the distance distribution of
any MRD code is determined by its parameters.

\begin{corollary}\label{weight1}
Let $\mC \subseteq \ma$ be an MRD code with $|\mC|\ge 2$ and minimum distance $d$.
We have $D_0(\mC)=1$, $D_i(\mC)=0$ for $1 \le i \le d-1$, and 
\begin{equation*} \label{questo}
D_i(\mC)=\sum_{u=0}^{d-1} {(-1)}^{i-u} q^{\binom{i-u}{2}}
\qbin{k}{i}\qbin{i}{u} +   \sum_{u=d}^i {(-1)}^{i-u} q^{\binom{i-u}{2}+m(u-d+1)}
\qbin{k}{i}\qbin{i}{u}
\end{equation*}
for $d \le i \le k$. 
\end{corollary}

\begin{proof}
Fix an $i$ with $d \le i \le k$. For $N \in \mC$ define $\mC-N=\{M-N \ | \ M \in \mC\}$. By definition
of distance distribution we have  
$$
|\mC| \cdot D_i(\mC)  = |\{(M,N) \in \mC^2 \ | \ \mbox{rk}(M-N)=i\}| = 
\sum_{N \in \mC} W_i(\mC-N).
$$
For all $N \in \mC$ the code $\mC-N$ is MRD. Moreover, 
$0 \in \mC-N$. The result now easily follows from Theorem~\ref{weight0}. 
\end{proof}

Corollary \ref{weight1} shows in particular that the weight distribution
of a linear MRD code is determined by $k$, $m$ and $d(\mC)$.
Recall from Proposition \ref{carMRD} that  an MRD code $\mC \subseteq \ma$ is characterized by
the property $d(\mC)+d(\mC^\perp)=k+2$. We now prove  
that 
the weight distribution of a linear code $\mC$ with 
$d(\mC)+d(\mC^\perp)=k+1$ is determined by 
$k$, $m$ and $\dim(\mC)$.
The following result is \cite[Corollary 28]{paperdistr}. 

\begin{corollary} \label{coroqmrd}
Let $\mC \subseteq \ma$ be a linear rank-metric code with $1 \le \dim(\mC) \le km-1$
and $d(\mC)+d(\mC^\perp)=k+1$. Then 
$$\dim(\mC) \not\equiv 0 \mod m \ \ \ \ \ \ \ \ \ \mbox{and} \ \ \ \ \ \  \ \ \ 
d(\mC)=k- \lceil \dim(\mC)/m\rceil +1.$$
Moreover, for all $d \le i \le k$ we have 
$$W_i(\mC) = \qbin{k}{i} \ \sum_{u=0}^{i-d(\mC)} {(-1)}^u q^{\binom{u}{2}} \qbin{i}{u}
\left( q^{\dim(\mC)-m(k+u-i)} -1\right).$$
\end{corollary}

\begin{proof}
Assume by contradiction that $\dim(\mC)=\alpha m$ for some $\alpha$. 
Applying Theorem~\ref{singbound} to $\mC$ and $\mC^\perp$ we obtain
\begin{equation} \label{twoin}
d(\mC) \le k-\alpha+1, \ \ \ \ \ \ \  d(\mC^\perp) \le \alpha+1.
\end{equation}
By Proposition \ref{carMRD}, the two inequalities in (\ref{twoin})
 are either both equalities,
or both strict inequalities. Since
$d(\mC)+d(\mC^\perp)=k+1$ by assumption, they must be both strict inequalities. Therefore 
$$d(\mC) \le k-\alpha, \ \ \ \ \ \ \  d(\mC^\perp) \le \alpha,$$
hence $d(\mC)+d(\mC^\perp) \le k$, a contradiction. This shows that 
$\dim(\mC) \not\equiv 0 \mod m$.

Now write $\dim(\mC)=\alpha m + \beta$ with $1 \le \beta \le m-1$.
Applying again Theorem \ref{singbound} to $\mC$ and 
$\mC^\perp$ one finds 
$$d(\mC) \le k - \left\lceil \frac{\alpha m + \beta}{m} \right\rceil +1= k-\alpha, \ \ \ \ \ \ \  
d(\mC^\perp) \le k- \left\lceil \frac{km-\alpha m -\beta}{m} \right\rceil = \alpha+1.$$
Since $d(\mC)+d(\mC^\perp)=k+1$, we must have 
$$d(\mC)= k- \left\lceil \frac{\alpha m + \beta}{m} \right\rceil +1= k- \left\lceil \frac{\dim(\mC)}{m} \right\rceil +1,$$
as claimed. The last part of the statement follows from Theorem \ref{04-02-15}.
\end{proof}

\section{Rank-metric anticodes} \label{sec:anti}

This section is devoted to rank-metric anticodes, i.e., 
rank-metric codes in which the distance between any two matrices is bounded
from above by a given integer $\delta$.

In Theorem \ref{boundanti} we give a bound for the cardinality of a (possibly non-linear) anticode,
using a  code-anticode-type bound. We also characterize 
optimal anticodes in terms of MRD codes.
Then we show that the dual of an optimal linear anticode is an optimal linear anticode.
The main results of this section appear in \cite{albrank} and \cite{tesialberto}.

 \begin{definition} \label{defanticode}
  Let $0 \le \delta \le k$ be an integer. A (\textbf{rank-metric}) \textbf{$\delta$-anticode}
  is a non-empty subset $\mA \subseteq \ma$ such that
  $d(M,N) \le \delta$ for all $M,N \in \mA$. We say that $\mA$ is 
  \textbf{linear} if it is an $\F_q$-linear subspace of $\ma$.
 \end{definition}
 
 \begin{example}\label{exanti}
 Any $\mA \subseteq \ma$ with $|\mA|=1$ is a $0$-anticode.
 The ambient space $\ma$ is a $k$-anticode.
 The vector space of $k \times m$ matrices over $\F_q$ whose last $k-\delta$ rows are zero 
 is a linear $\delta$-anticode of dimension ${m\delta}$.
 \end{example}
 
In the sequel we work with a fixed integer 
$0 \le \delta \le k$. Moreover, for $\mA,\mC \subseteq \ma$ we set 
 $\mA + \mC =\{M+N \ | \ M \in \mA, \ N \in \mC\}$.

 \begin{theorem} \label{boundanti}
Let $\mA \subseteq \ma$ be a $\delta$-anticode. Then 
 $|\mA| \le q^{m\delta}$. Moreover, if $\delta \le k-1$ then the following are equivalent.
 \begin{enumerate}
 \item  $|\mA| = q^{m\delta}$. \label{ch1}
  \item $\mA+\mC =\ma$ for some MRD code $\mC$ with $d(\mC)=\delta+1$. \label{ch2}
 \item $\mA+\mC =\ma$ for all MRD codes $\mC$ with $d(\mC)=\delta+1$. \label{ch3}
 \end{enumerate}
  \end{theorem}
  
  \begin{proof}
  Let
  $\mC \subseteq \ma$ be any MRD code with $d(\mC)=\delta+1$.
  Such a code exists by Theorem \ref{attaingabid}.
  For all $M \in \mA$ let $[M]=M+\mC= \{M+N \ | \ N \in \mC\}$.
  Then $[M] \cap [M'] = \emptyset$ for all 
  $M,M' \in \mA$ with $M \neq M'$. Moreover, by definition of 
  MRD code we have
  $|[M]|=|\mC|=q^{m(k-\delta)}$ for all $M \in \mA$, hence
  $$|\ma| \ge \left| \bigcup_{M \in \mA} [M]\right| = \sum_{M \in \mA} |[M]| = |\mA| \cdot |\mC|
  = |\mA| \cdot q^{m(k-\delta)}.$$
  Therefore $|\mA| \le q^{m\delta}$, and equality holds if and only if 
  $$\ma = \bigcup_{M \in \mA} [M] = \mA + \mC.$$
 A similar argument shows that properties \ref{ch1}, \ref{ch2} and \ref{ch3} are equivalent.
  \end{proof}

\begin{definition} \label{defopt}
We say that a $\delta$-anticode $\mA$ is ({cardinality})-\textbf{optimal} if it attains the bound of Theorem \ref{boundanti}.
\end{definition}

\begin{remark}
Example \ref{exanti} shows the existence of optimal linear $\delta$-anticodes for all choices of the parameter $\delta$. 
\end{remark}

In the remainder of the section we prove that the dual of an optimal linear $\delta$-anticode 
is an optimal $(k-\delta)$-anticode. The result may be regarded as the analogue of Theorem
\ref{dualofis} in the context of rank-metric anticodes.
We start with a preliminary result on the weight distribution of MRD codes.

\begin{lemma} \label{prel}
 Let $\mC \subseteq \ma$ be an MRD code
 with $0 \in \mC$, $|\mC| \ge 2$ and $d(\mC)=d$. Then
$W_{d+\ell}(\mC)>0$ for all $0 \le \ell \le k-d$.
\end{lemma}

\begin{proof}
 By
Theorem \ref{weight0}, we shall prove the lemma for a given  MRD
code $\mC
\subseteq \ma$ of our choice with $|\mC| \ge 2$, minimum distance
$d$, and $0 \in \mC$. We will first produce a convenient MRD code with the prescribed
properties.

Let $C \subseteq \F_{q^m}^k$ be the vector rank-metric code constructed in the
proof of Theorem~\ref{attaingabid},  with evaluation set
$E=\{\beta_1,...,\beta_k\}$ and evaluation map $\mbox{ev}_E$. Let $\Gamma$ be a
basis of
$\F_{q^m}$ over $\F_q$. By Proposition~\ref{gabtodel}, the set 
$\mC=\Gamma(C) \subseteq \ma$ is a linear 
code with $\dim (\mC)=m(k-d+1)$ and the same weight distribution as $C$. 
In particular, $\mC$ is a non-zero linear MRD code
of minimum distance $d$.

Now we prove the lemma for the MRD code $\mC$ constructed above. Fix $\ell$ with 
$0 \le \ell \le k-d$. Define $t=k-d-\ell$, and let $U
\subseteq \F_{q^m}$ be the $\F_q$-subspace generated by $\{
\beta_1,...,\beta_{t}\}$. If $t=0$ we set $U$ to be the
zero space.
By \cite{niede}, Theorem 3.52, 
$$p_U= \prod_{\gamma \in U} (x-\gamma)$$
is a linearized polynomial over $\F_{q^m}$ of degree $t=k-d-\ell \le k-d$,
i.e., 
 $p_U \in  \mbox{Lin}_q(n,k-d)$.  Therefore by Proposition \ref{gabtodel} it
suffices to prove that $\mbox{ev}_E(p_U)=(p_U(\beta_1),...,p_U(\beta_k))$
has rank
$d+\ell=k-t$. Clearly, $V(p_U)=U$. In particular we have
$\mbox{ev}_E(p_U)=(0,...,0,p_U(\beta_{t+1}),...,p_U(\beta_k))$. We will
show that 
$p_U(\beta_{t+1}),..., p_U(\beta_k)$ are linearly independent over $\F_q$.
Assume that there exist $a_{t+1},...,a_k \in \F_q$ with 
$\sum_{i=t+1}^k a_i p_U(\beta_i)=0$. Then we have $p_U \left(
\sum_{i=t+1}^k
a_i\beta_i \right)=0$, i.e.,  $\sum_{i=t+1}^k
a_i\beta_i \in V(p_U)=U$.
It follows that there exist $a_1,...,a_t \in \F_q$ such that
$\sum_{i=1}^t a_i \beta_i = \sum_{i=t+1}^k a_i\beta_i$, i.e.,
$\sum_{i=1}^t
a_i \beta_i - \sum_{i=t+1}^k a_i\beta_i=0$. Since $\beta_1,...,\beta_k$
are
 independent over $\F_q$, we have $a_i=0$ for all $i=1,...,k$. In
particular $a_i=0$ for $i=t+1,...,k$. Hence
$p_U(\beta_{t+1}),...,p_U(\beta_k)$ are linearly independent over $\F_q$, as claimed.
\end{proof}

The following proposition characterizes optimal linear anticodes in terms of their intersection with
 linear MRD codes.

\begin{proposition} \label{criterio}
 Assume $0 \le \delta \le k-1$, and let $\mA \subseteq
\ma$ be a linear code with $\dim (\mC) =m\delta$. The following are equivalent.
\begin{enumerate}
 \item $\mA$ is an optimal $\delta$-anticode.
 \item $\mA \cap \mC =\{0\}$ for all non-zero MRD linear codes $\mC \subseteq
\ma$ with $d(\mC)=\delta+1$.
\end{enumerate}
\end{proposition}

\begin{proof}
By Theorem \ref{boundanti}, it suffices to show that 
if $\mA \cap \mC =\{0\}$ for all non-zero MRD linear codes $\mC \subseteq
\ma$ with $d(\mC)=\delta+1$, then $\mA$ is a $\delta$-anticode.

By
contradiction, assume that $\mA$ is not a $\delta$-anticode. 
Since $\mA$ is linear, by definition of $\delta$-anticode there exists 
$N \in \mC$ with
$\rk(N) \ge \delta+1$. Let $\mD$ be a non-zero linear MRD code with
$d(\mD)=\delta+1$ (see Theorem \ref{attaingabid} for the existence of such a code). 
By Lemma \ref{prel} there exists $M \in \mD$ with
$\rk(M)=\rk(N)$. There exist invertible matrices $A$ and $B$ of size $k \times
k$ and $m\times m$, resp., such that $N=AMB$. Define $\mC=A\mD B= \{
APB \ | \  P \in \mD \}.$ Then $\mC \subseteq \ma$ is a
non-zero linear MRD
code with $d(\mC)=\delta+1$ and such that $N \in \mA \cap \mC$. Since
$\rk(N) \ge \delta+1 \ge 1$, $N$ is not the zero matrix. Therefore 
$\mA \cap \mC \neq \{0\}$, a contradiction.
\end{proof}

We conclude the section showing that the dual of an optimal linear anticode is an optimal
linear anticode.

\begin{theorem} \label{dualanticode}
Let  $\mA \subseteq \ma$
be an optimal linear $\delta$-anticode. Then 
$\mA^\perp$ is an optimal linear $(k-\delta)$-anticode.
\end{theorem}
\begin{proof}
 Let $\mA \subseteq \ma$ be an optimal linear $\delta$-anticode.
 If $\delta=k$ then the result is trivial.  From now on
we assume $0
\le \delta \le k-1$. By Definition \ref{defopt} we have $\dim (\mA) = m\delta$,
hence $\dim (\mA^\perp)=m(k-\delta)$. Therefore by Proposition~\ref{criterio} 
it suffices to show that 
 $\mA^\perp \cap \mC = \{0\}$ for all non-zero linear MRD codes
$\mC \subseteq \ma$ with $d(\mC)=k-\delta+1$. Let
$\mC$ be such a code. Then 
$$\dim (\mC)=m(k-(k-\delta+1)+1)=m\delta<mk.$$ Combining Theorem \ref{dualofis} and 
Proposition~\ref{carMRD} one shows that
$\mC^\perp$ is a linear MRD code with $d(\mC^\perp)=k-(k-\delta+1)+2=\delta+1$.
By Proposition \ref{criterio} we have $\mA \cap \mC^\perp = \{0\}$. Since
$\dim (\mA)+\dim (\mC^\perp)=m\delta+m(k-(\delta+1)+1)=mk$, we have $\mA
\oplus \mC^\perp = \ma$. Therefore $\{0\} = (\ma)^\perp =(\mA \oplus \mC^\perp)^\perp = \mA^\perp \cap \mC$. This shows the theorem.
\end{proof}

\end{document}